\documentclass[runningheads]{llncs}
\usepackage{float}
\usepackage{graphicx,hyperref}
\usepackage[ruled,linesnumbered]{algorithm2e}
\usepackage{subfigure}
\usepackage{stmaryrd}
\usepackage{amsmath}
\usepackage{amssymb,bm,mathrsfs}
\usepackage{latexsym}
\usepackage{color}
\usepackage{lineno} %
\usepackage{setspace}
\usepackage{xcolor}
\colorlet{T}{black}

\newcommand{\ignore}[1]{}

\newcommand{\R}{\mathbb{R}}
\newcommand{\Rn}{\mathbb{R}^n}

\newcommand{\secure}{{\em secure}}
\newcommand{\compromised}{{\em compromised}}

\newtheorem{control-rule}{Event-based Trigger Rule}
\begin{document}
\title{An Event-based Parameter Switching Method for Controlling Cybersecurity Dynamics}
\titlerunning{An Event-based Method for Controlling Cybersecurity Dynamics}
%
\author{Zhaofeng Liu\inst{1} \and
Wenlian Lu\inst{1,2,3} \and
Yingying Lang\inst{1}}
\authorrunning{Z. Liu et al.}
%
\institute{School of Mathematical Sciences, Fudan University, Shanghai 200433, China
\email{zhaofengliu@hotmail.com}\\
\email{wenlian@fudan.edu.cn}\\
\email{langyy18@fudan.edu.cn}
\\
\and
Shanghai Center for Mathematical Sciences, Fudan University, Shanghai 200433, China
\\
\and
Shanghai Key Laboratory for Contemporary Applied Mathematics, Shanghai 200433, China
}
\maketitle              
\begin{abstract}
  This paper proposes a new event-based parameter switching method for the control tasks of cybersecurity in the context of preventive and reactive cyber defense dynamics. Our parameter switching method helps avoid excessive control costs as well as guarantees the dynamics to converge as our desired speed. Meanwhile, it can be proved that this approach is Zeno-free. A new estimation method with adaptive time windows is used to bridge the gap between the probability state and the sampling state. With the new estimation method, several practical experiments are given afterwards.


\end{abstract}
\section{Introduction}
With the rapid development of internet technology (IT) and internet of things (IoT) technology, cybersecurity issues are attached more and more importance these years due to people's increasing reliance on internet.
About 2 million to 5 million computers worldwide suffered from malware (are infected with malware) every day, estimating by the domestic network security company Jinshan in 2016. Cybersecurity, as a new subject for research, has received extensive attention from the academic community.

The emerging research field 'cybersecurity dynamics' \cite{XuCybersecurityDynamicsHotSoS2014,XuBookChapterCD2019} is an interdisciplinary field,
conceived from the methodology of several early studies in {\em biological epidemiology} (e.g., \cite{McKendrick1926,Kermack1927,Bailey1975,Anderson1991,HethcoteSIAMRew00}) and its variants in
{\em cyber epidemiology} (e.g.,  \cite{KephartOkland91,KephartOkland93,Pastor2001,Moreno2002}),
{\em interacting particle systems} \cite{Liggett1985},
and {\em microfoundation in economics} \cite{Hoover2010}.
Different from the classical researches oriented to specific tools, such as Cryptography and Database Security, cybersecurity dynamics studies the  offensive and defensive models under various circumstances from a whole-network perspective.

\subsection{Our Contributions}
\label{sec:our-contributions}

In this paper, we investigate how to control the evolution of cybersecurity dynamics more efficiently and effectively in the context of
{\em preventive and reactive cyber defense dynamics}, and guarantee its globally convergence to a safe state.

When it comes to the control problem of the highly nonlinear network dynamics system, the traditional method is to use a single control strategy, that is, to adjust the dynamic parameters 
and maintain them, so that the security dynamics of the network space converges to a safe state globally.
However, maintaining a high level of dynamic defense strategy may result in high control costs.
Such excessive prevention and control will lead to a waste of control resources to some extent, and may even have negative effects on the stable operation of the network systems.

In order to solve this problem, this paper proposes an event-based parameter switching approach to save control resources and control the evolution of cybersecurity dynamics in a decentralized manner. It is also proved that this approach is Zeno-free, that is, it will not fall victim to the Zeno behavior. Numerical examples show that the maintenance hours of high-cost control strategies can be reduced by more than $40\%$  with our parameter switching approach.

In addition, this paper provides an estimation method to bridge the gap between the probability state and the sampling state when using the method in practice.
Different from the equilibrium state estimation problem\cite{liu2020using}, the event-based control method requires considering the timeliness of probability estimation. A new adaptive estimation method is proposed to verify the effectiveness of the event-based parameter switching method in the control process through numerical examples.


\

\subsection{Related Work}
\label{sec:related-work}

Similar event-based methods have been employed in many other application settings before (see, for example, \cite{Astrom2002,Tabuada2007,6425820}).
In practical application, one of the essential problems is that this method should not fall victim to the Zeno behavior, which can lead to infinitely many events within a finite period of time, thus invalidate the method \cite{Johansson1999}.
The importance of assuring Zeno-freeness in event-triggered control is witnessed by a body of literature, including
\cite{Tabuada2007,Johannesson2007,WangLemmon2011,Dimarogonas2012}. 

The first cybersecurity dynamics model was proposed in \cite{XuAINA07}. Ten years later,
\cite{ZhengSub} demonstrated that a certain class of cybersecurity dynamics is globally convergent in the entire parameter universe, which laid a foundation for our further research. The notion of cybersecurity dynamics, as discussed in \cite{XuCybersecurityDynamicsHotSoS2014,XuBookChapterCD2019}, has opened the door to a new research field.

%

The preventive and reactive cyber defense dynamics is a particular kind of cybersecurity dynamics.
Several other kinds of cybersecurity dynamics have been demonstrated in early studies, such as the models aiming to accommodate adaptive defenses \cite{XuTAAS2014}, active defenses \cite{XuHotSoS15,XuInternetMath2015ACD,XuGameSec13}, and
proactive defenses \cite{XuHotSOS14-MTD}.
It is worth mentioning that the event-based parameter switching method may be extended and applied to the various kinds of dynamics.

\subsection{Paper Outline}
In Section \ref{sec:model}, we briefly review the preventive and reactive cyber defense dynamics model and its global convergence in the entire parameter universe \cite{ZhengSub}. And then we state the problems we addressed in this paper.
In Sections \ref{sec:proof-of-problem},
we present an
event-based parameter switching method for controlling the cybersecurity dynamics, and prove the effectiveness of the control method (without Zeno behavior).
In Section \ref{sec:simulation}, we present numerical examples on the theoretical model. Then in Section \ref{sec:gap}, we show how to apply this event-based parameter switching method in practice by bridging the gap between the probability-state in the theoretical model and the sample-state in practice, using an stochastic process method with adaptive time windows.
In Section \ref{sec:conclusion}, we conclude the paper with open problems.

\section{Problem Statement}
\label{sec:model}

\subsection{Review of Preventive and Reactive Defense Dynamics}
As a particular kind of cybersecurity dynamics, the preventive and reactive defense dynamics model is first introduced in \cite{XuAINA07},
and the convergence properties of the dynamics is studied in \cite{XuTAAS2012} .
Later \cite{ZhengSub} fully analyzed the convergence issues, not only considered the common situation with node homogeneity (i.e., the parameters are node-independent), but also a more general situation with node heterogeneity (i.e., the parameters are nodes-dependent).
The paper proved that this dynamics model is globally convergent in the entire parameter universe, that is, there is always a unique equilibrium, whose exact value (or position) depends on the specific parameter values instead of the initial state of the dynamics.

In the preventive and reactive defense dynamics model, we consider two classes of defenses: {\em preventive defenses} and {\em reactive defenses}, and two classes of attacks: {\em push-based attacks} and {\em pull-based attacks}.

Suppose that the attack-defense interaction occurs over an {\em attack-defense} graph structure $G=(V,E)$, where $V$ is the vertex set representing computers and $(u,v)\in E$ means computer $u$ can directly attack  computer $v$ using push-based attack strategy  (i.e., the communication from $u$ to $v$ is allowed by the security policy).
 $G$ can be derived from the security policy of a networked system and the physical network in question.
Without loss of generality, we do {\em not} make  any  restrictions on the structure of $G$ (e.g., $G$
may be directed or undirected).
Denote the adjacency matrix of $G$ by $A=[a_{vu}]_{n\times n}$, where $a_{vu}=1$ if and only if $(u,v) \in E$.
Since the model aims to describe the attacks between computers,
we set $a_{vv}=0$.
Let $N_{v}=\{u\in V:~(u,v)\in E\}$.

In this paper, we consider the continuous-time model described in \cite{ZhengSub}.
At any time point, a node $v \in V$ is in one of two states: ``0'' means  \secure\ but vulnerable,
or ``1'' means  \compromised.
Let $s_v(t)$ and $i_v(t)$ be the probability that $v$ is {\em secure} and {\em compromised} at time $t$ respectively. It is obvious that $s_v(t)+i_v(t)=1$, $s_v(t)$ and $i_v(t)$ explain the term {\em probability-state}.



For a node $v\in V$ at time $t$, let $\theta_{v,1\to0}(t)$ abstract the effectiveness of the reactive defenses and $\theta_{v,0\to1}(t)$ abstract the capability of attacks against the preventive defenses.
$\beta_v\in (0,1]$ represents the probability that the \compromised\ computer $v$ changes to the \secure\ state because the attacks are detected and cleaned up by the reactive defenses.
Then, $\theta_{v,1\to0}(t)=\beta_v$.
Let $\alpha_v\in [0,1]$ denote the probability that the
\secure\ computer $v$ becomes \compromised\ despite the presence of the preventive defenses (i.e., the preventive defenses are penetrated by the pull-based attacks). And
let $\gamma_{uv}\in (0,1]$ denote the probability that a
\compromised\ computer $u$ wages a successful attack against a \secure\ computer $v$ despite the preventive defenses (i.e., the preventive defenses are penetrated by push-based attacks), where $(u,v)\in E$.
Under the assumption that the attacks are waged independent of each other, it holds that
\begin{equation}
\label{eq:nonlinear-term}
\theta_{v,0\to1}(t)=1-(1-\alpha_v)\prod_{u\in N_{v}}\big(1-\gamma_{uv} i_{u}(t)\big).
\end{equation}

The dynamics can be rewritten as a system of $n$ nonlinear equations for $v\in V$ \cite{ZhengSub}:
\begin{equation}
\begin{aligned}
\label{Main}
\dfrac{{\rm d}i_{v}(t)}{{\rm d}t}=f_{v}(i)=-\beta_v i_{v}(t)
+\bigg[1-(1-\alpha_v)\prod_{u\in N_{v}}\big(1-\gamma_{uv} i_{u}(t)\big)\bigg]\big(1-i_{v}(t)\big).
\end{aligned}
\end{equation}

Notice that system \eqref{Main} is globally stable (i.e., there exists a unique equilibrium $i^{*}\in [0,1]^n$ such that every trajectory of system \eqref{Main} converges to $i^{*}$)  no matter whether the parameters are nodes-dependent or nodes-independent \cite{ZhengSub}. If the parameters of the network system are nodes-independent (i.e., $\alpha_v=\alpha$, $\beta_v=\beta$ for any $v\in V$ and $\gamma_{uv}=\gamma$ for any $u,v\in V$, $(u,v)\in E$), the global convergence of system \eqref{Main} can be summarized as follows:


\begin{itemize}
\item If the attacker wages both push-based and pull-based attacks on some nodes $v\in V$ (i.e., $\alpha_v>0$ for some nodes $v\in V$),
system \eqref{Main} is globally convergent in the entire parameter universe and the dynamics converges to a unique {\em nonzero} equilibrium exponentially.
\item If the attacker only wages push-based attacks (i.e., $\alpha_v=0$ for any nodes $v\in V$),
system \eqref{Main} is still globally convergent in the parameter universe but the convergence speed depends on all the model parameters $(\beta_v,\gamma_{uv})$ and
the largest eigenvalue $\lambda_{A,1}$ of adjacency matrix $A$.
\end{itemize}

In this paper, we need to control the dynamics with nodes-dependent parameters converging to equilibrium zero. Despite the complexity of nodes heterogeneity, we can still take advantage of the convergence properties of the dynamics with nodes-independent parameters.

\subsection{Problem Statement: Controlling Cybersecurity Dynamics}

In this paper, we focus on controlling the convergence process of the cybersecurity dynamics model in a decentralized control manner, which means that we only need to observe the state of the target node $v$ during the control process of $v$, with no need to observe the states of its neighbors within the network. For every node $v$ of the network, we need to control $i_v(t)$ converging to zero at our target convergence speed with relatively low control cost by switching its parameter $\beta_v$ 
according to our control rule.

Before presenting the control method, we need to finish two pre-control steps to assure the effectiveness of our method. As introduced before, we need to control the preventive and reactive defense dynamics converging to equilibrium zero globally. So firstly, as discussed above, we need to force the parameter $\alpha_v=0$ for each node $v$, which means that the threats of pull-based attacks are eliminated after the first step of the control process (e.g., connections between some compromised websites and the network system in which the pull-based epidemic spreading takes place are all cut off). This is the first step of the pre-control process. Then the corresponding push-based dynamics model we focus on can be rewritten as:
\begin{equation}
\begin{aligned}
\label{Main-push}
\dfrac{{\rm d}i_{v}(t)}{{\rm d}t}=f_{v}(i)=-\beta_v i_{v}(t)
+\bigg[1-\prod_{u\in N_{v}}\big(1-\gamma_{uv} i_{u}(t)\big)\bigg]\big(1-i_{v}(t)\big).
\end{aligned}
\end{equation}

As for $\beta_v$, we select two reactive defense strategy with different control cost for our parameter switching method, including one relatively strict defense strategy with higher control cost, denoted by $\beta_{+}$, and one relatively relaxed defense strategy with lower control cost, denoted by $\beta_{-}$. Then we apparently have $\beta_{+}>\beta_{-}$. These two strategies should satisfy the conditions that they are both able to make the dynamics converge to equilibrium zero. The key difference between them is, comparing to our target convergence speed with regard to the dynamics to be controlled, the dynamics with $\beta_v=\beta_{+}$ should converge faster than the target speed for all nodes in $V$, while the dynamics with $\beta_v=\beta_{-}$ may converge more slowly than the target speed.

The classical approach to control the dynamics through adjusting the reactive defense strategy is the trivial method that forcing $\beta_{v}=\beta_{+}$ during the entire control process, which is inefficient because it may cost too many defense resources to maintain the relatively strict reactive defense strategy. Besides, the convergence speed of the dynamics under high-cost control may be faster than what we actually need, causing redundancies and wastes of defense resources to some extent.

As discussed above, both the relatively strict defense strategy and the relatively relaxed defense strategy need to be able to make the dynamics converge to equilibrium zero. So a safe method to get equilibrium zero is to let $\beta_v/\gamma_{max}\ge\lambda_{A,1}$ for both $\beta_v=\beta_{+}$ and $\beta_v=\beta_{-}$, where $\gamma_{max}$ denotes the maximum value of probability $\gamma_{uv}$ for all neighbor nodes pair $(u,v)\in E$. But due to the variety of parameter $\gamma_{uv}$, $\gamma_{max}$ can be relatively large in practice, leaving little choice for parameter $\beta_v$. Therefore, the second step of the pre-control process is to force $\gamma_{max}$ to a relatively small value, which means we need to permanently reinforce the preventive defense strategy for the nodes which are more vulnerable to push-based attacks launched by the attacker (e.g., a stronger network firewall or filter is deployed). 

After the two pre-control steps, in which all the pull-based attacks have been eliminated and $\gamma_{max}$ is relatively small, we now employ an parameter switching method to control the convergence speed
with relatively low cost of defense resources. We use an event-based mechanism to define the cost-saving parameter switching rule.
\subsection{Notations}
Table \ref{table:model-variables-and-parameters} summarizes the major notations used in the paper.


\begin{table}[!htbp]
\centering
\caption{Notations used throughout the paper.\label{table:model-variables-and-parameters}}
\begin{tabular}{|c|p{0.8\textwidth}|}
\hline
$I_n$ & the $n*n$ identity matrix\\ \hline
$\R$ & the set of real numbers\\ \hline
$\mathbb{N}$ & the set of positive integers and zero\\ \hline
$\|i\|_{1}$ &  $\|i\|_{1}=\sum_{v=1}^{n}\|i_{v}\|$ is the $l_{1}$-norm for an $n$-dimensional vector $i=[i_{1},\ldots,i_{n}]\in\Rn$.
Note that the result equally holds with respect to other norms.\\ \hline
$G=(V,E), A$ & the attack-defense graph structure $G$ with adjacency matrix $A=[a_{vu}]_{n\times n}$ where $a_{vu}=1$ if and only if $(u,v)\in E$\\ \hline
$N_{v}$ & $N_{v}=\{u\in V:(u,v)\in E\}$\\ \hline
$\alpha_v \in [0,1]$ & the probability that \secure\ node $v$ becomes \compromised\ because pull-based attack penetrates preventive defense \\ \hline 
$\beta_v \in (0,1]$ & the probability that \compromised\ node $v$ becomes \secure\ because reactive defense detects and cleans compromise \\ \hline 
$\gamma_{uv} \in (0,1]$ & the probability that a \compromised\ neighbor node $u$ wages a successful push-based attack against \secure\ node $v$ \\ \hline
$\gamma_{max} \in (0,1]$ & the maximum value of probability $\gamma_{uv}$ for all neighbor nodes pair $(u,v)\in E$ \\ \hline
$\beta_{-}$, $\beta_{+}$ & the value of parameter $\beta_v$ for all nodes $v\in V$ in low-cost (high-cost) reactive defense setting\\ \hline
$i_v(t)$, $i(t)$ & the probability $v$ is in \compromised\ state at time $t$; $i(t)=[i_{1}(t),\cdots,i_{n}(t)]$ \\ \hline
$\varphi_{up}(s)$, $\varphi_{up}(s)$ & the decision functions that trigger high-cost (low-cost) control events\\ \hline
$t_{k}^{v}$, $\tau_{k}^{v}$ & the time for the $k$-th high-cost (low-cost) control event at $v\in V$ in the event-based control method; $t_{1}^{v}=\tau_{0}^{v}=0$  \\ \hline
$T_{-}^v$, $T_{+}^v$ & the total time of maintaining the parameter $\beta_v$ of the target node in low-cost (high-cost) reactive defense setting during control process\\ \hline
$\mathcal{S}(t)$, $\mathcal{S}_{v}(t)$ & the exponential speed index of the convergence speed of $i(t)$ ($i_v(t)$) when the dynamics converge exponentially\\ \hline
$\chi_v(t)$ & the sample-state of node $v$ at time $t$; 0 means {\em secure} and 1 means {\em compromised} \\ \hline
$\widehat{i_{v}(t)}$, $\widehat{s_{v}(t)}$ & the probability $v$ is in the \compromised\ (\secure) state at time $t$ as estimated from the sample-states \\ \hline
$\mathcal{W}$, $\mathcal{W}^{'}(t)$ & the time window with fixed (adaptive) time length for estimating the probability-states from sample-states\\ \hline

\end{tabular}
\end{table}

\section{An Event-based Parameter Switching Method}
\label{sec:proof-of-problem}

In this section, we initiate an event-based parameter switching method to control the convergence speed to our target speed, with relatively low cost of defense resources. That is to say, we switch the reactive defense strategy between two predefined settings in practice, according to the event-based trigger rule we proposed. We first show that the global dynamics under control will converge to equilibrium zero at our target convergence speed. Then we prove that there is no Zeno behavior during the entire control process.

\subsection{Designing Event-based Parameter Switching Rule}
\label{subsec:Design-Rule}

We apply the control method on all nodes of the network system. But for the purpose of clarification, we focus on one target node $v$ to explain the control process corresponding to the decentralized control manner as discussed above.

We switch the parameter $\beta_v$ of the target node $v$ between two different groups of parameters alternately, which are the low-cost parameter $\beta_v=\beta_{-}$ and the high-cost parameter $\beta_v=\beta_{+}$. Here the footnotes 'lc' stands for 'low cost' and 'hc' stands for 'high cost'. We use $T_{-}^v$ and $T_{+}^v$ respectively to denote the total time of maintaining the two parameters $\beta_{-}$ and $\beta_{+}$ during the entire control process for node $v\in V$. We want the dynamics to converge to zero at target convergence speed with a relatively low total cost, which means to make the mean of time ratio $\frac{1}{n}\sum_{v\in V}\frac{T_{+}^v}{T_{+}^v+T_{-}^v}$ relatively small.

Before defining the event-trigger rule, let us review these definition and lemma.

For all nodes $v\in V$, let $n*n$ real matrix $K=\{\gamma_{uv}\}_{u,v=1}^n$, where $\gamma_{vv}=0$ for all nodes $v\in V$, and let non-singular $n*n$ real diagonal matrix $B=diag(\{\beta_v\}_{v=1}^n)$. Notice that $\beta_v$ takes values in $\{\beta_{-},\beta_{+}\}$ for all nodes $v\in V$. Let $B_{+}=\beta_{+}I_n$ and $B_{-}=\beta_{-}I_n$, where $I_n$ is $n*n$ identity matrix. Then we have following lemma.
\begin{lemma}
\label{lemma-get-m-matrix}
If there exists some positive constant number $\iota$, so that $J=B_{+}-K-\iota I_n$ is a monotone matrix (M-matrix),
then there exists a positive diagonal matrix $P=diag(\{p_v\}_{v=1}^n)$ such that $[(-J)P+P^T(-J)]$ is negative definite.
\end{lemma}


Apparently, our target convergence speed should be faster than the dynamics with $\beta_v=\beta_{-}$ for all nodes in $V$, and slower than the dynamics with $\beta_v=\beta_{+}$ for all nodes in $V$. We can set our target convergence speed as $C{\rm e}^{-\iota t}$ as long as $\iota>0$ satisfies Lemma \ref{lemma-get-m-matrix}, where $C$ is a positive constant number. We will prove the effectiveness later.

In our event-based parameter switching method, we switch the parameter $\beta_v$ 
when an event is triggered (e.g. certain conditions are satisfied). Between two consecutive events, the parameter $\beta_v$ holds. 
We define two criterion functions according to our target speed of the convergence process:
\begin{eqnarray*}
\begin{cases}
\displaystyle
\varphi_{up}(t)&=~{\rm e}^{-\iota t},{\forall}t\ge0,\\[8pt]
\displaystyle
\varphi_{low}(t)&=~L*{\rm e}^{-\iota t},{\forall}t\ge0.
\end{cases} 
\end{eqnarray*}
where $L$ is an positive constant number satisfying $0<L<1$. Notice that $\varphi_{up}(t)$ represents the ideal convergence process at our target convergence speed. Theoretically, $\varphi_{up}(t)$ can be a polynomial function if the original dynamics converge polynomially. But from the Sard's Lemma in \cite{sard1942measure}, the parameter regime that causing polynomial convergence speed is in a zero measure set and cannot be chosen in practice. So without loss of generality, we let $\varphi_{up}(t)$ be an exponential function for the simplification of narrative. It is also worth noting that $\varphi_{low}(t)$ can be defined in other functional form (only need to satisfy the conditions of convergence speed).

With the matrix $P$ defined in Lemma \ref{lemma-get-m-matrix}, we now define rigger rule as follows:
\begin{definition}[event-based trigger rule]
\label{trigger-rule-control}
$P=diag(\{p_v\}_{v=1}^n)$ is as defined in Lemma \ref{lemma-get-m-matrix}.
Let $m_v(t)=p_v^{-1}i_v(t)$, then for $k=1,2,\ldots$, the trigger rule is defined as:
\begin{itemize}
\item if $m_v(0)\ge 1$, then let $t_{1}^{v}=0$, and
\begin{eqnarray*}
\begin{cases}
\displaystyle
\tau_{k}^{v}=&\inf\bigg\{s\ge t_{k}^{v}:m_v(s)\le \varphi_{low}(s)\bigg\}\\[8pt]
\displaystyle
t_{k+1}^{v}=&\inf\bigg\{s\ge \tau_{k}^{v}:m_v(s)\ge \varphi_{up}(s)\bigg\}
\end{cases} 
\end{eqnarray*}
\item if $m_v(0)<1$, then let $\tau_{0}^{v}=0$, and
\begin{eqnarray*}
\begin{cases}
\displaystyle
t_{k}^{v}=&\inf\bigg\{s\ge \tau_{k-1}^{v}:m_v(s)\ge \varphi_{up}(s)\bigg\}\\[8pt]
\displaystyle
\tau_{k}^{v}=&\inf\bigg\{s\ge t_{k}^{v}:m_v(s)\le \varphi_{low}(s)\bigg\}
\end{cases} 
\end{eqnarray*}
\end{itemize}
which specifies two sequences of parameter switching events:
\begin{itemize}
\item High-cost control event: At time $t_{k}^{v}$, the value of parameter $\beta_v$ switches to $\beta_{+}$ which generates relatively high control cost.
\item Low-cost control event: At time $\tau_{k}^{v}$, the value of parameter $\beta_v$ switches to $\beta_{-}$ which generates relatively low control cost.
\end{itemize}
\end{definition}


As discussed above, both $\beta_{-}$ and $\beta_{+}$ need to be able to make the dynamics converge to equilibrium zero in the parameter regime after pre-control. So system \eqref{Main-push} can be written as: for target node $v\in V$,
\begin{itemize}
\item If $t\in[t_{k}^{v},\tau_{k}^{v})$,
\begin{equation}
\begin{aligned}
\label{high-cost-dynamics}
\frac{{\rm d}i_{v}(t)}{{\rm d}t}
=-\beta_{+} i_{v}(t)
+\bigg[1-
\prod_{u\in N_{v}}\Big(1-\gamma_{uv} i_{u}(t)\Big)\bigg]\big(1-i_{v}(t)\big),
\end{aligned}
\end{equation}
\item If $t\in[\tau_{k}^{v},t_{k+1}^{v})$,
\begin{equation}
\begin{aligned}
\label{low-cost-dynamics}
\frac{{\rm d}i_{v}(t)}{{\rm d}t}
=-\beta_{-} i_{v}(t)
+\bigg[1-
\prod_{u\in N_{v}}\Big(1-\gamma_{uv} i_{u}(t)\Big)\bigg]\big(1-i_{v}(t)\big),
\end{aligned}
\end{equation}
\end{itemize}

For $k=1,2,\ldots$, we regard the two control steps \eqref{high-cost-dynamics} and \eqref{low-cost-dynamics} during $t\in[t_{k}^{v},t_{k+1}^{v})$ as the $k$-th {\em control cycle}.


\subsection{Analyzing the Event-based Parameter Switching Method}

Under the control procedure proposed above, we will prove the effectiveness of the event-based parameter switching method, that is, the new dynamics of the target node under control will converge to zero at our target convergence speed, and what's more important, with no Zeno behavior.

\begin{theorem}
\label{theorem-main}
For any node $v\in V$, system \eqref{high-cost-dynamics} \eqref{low-cost-dynamics} generated by the event-based parameter switching control strategy (trigger rule Definition \ref{trigger-rule-control}) will converge to zero at the convergence speed same as $\varphi_{up}(t)$, with no Zeno behavior.
\end{theorem}

\begin{proof}
We first prove that $m_v(t)=p_v^{-1}i_v(t)$ under parameter switching control strategy (trigger rule Definition \ref{trigger-rule-control}) could never exceed $\varphi_{up}(t)$ after time $\tau_{1}^{v}$ for all nodes $v\in V$ (i.e., it always holds that $p_v^{-1}i_v(t)\le\varphi_{up}(t)$ after time $\tau_{1}^{v}$).

For any node $v\in V$, with regard to its original dynamic \eqref{Main-push}, let $u_{-}$ be the smallest index in $N_{v}$. Notice that
\begin{align*}
    &~1-\prod_{u\in N_{v}}\Big(1-\gamma_{uv} i_{u}(t)\Big)\\
=&~\bigg[\Big(1-\gamma_{u_{-}v} i_{u_{-}}(t)\Big)+\gamma_{u_{-}v} i_{u_{-}}(t)\bigg]
-\prod_{u\in N_{v}}\Big(1-\gamma_{uv} i_{u}(t)\Big)\\
=&~\gamma_{u_{-}v} i_{u_{-}}(t)+\Big(1-\gamma_{u_{-}v} i_{u_{-}}(t)\Big)
\bigg[1-\prod_{u>u_{-},u\in N_{v}}\Big(1-\gamma_{uv} i_{u}(t)\Big)\bigg].
\end{align*}

This recurrent process will lead to
\begin{align*}
    \frac{{\rm d}}{{\rm d}t}i_{v}(t)
=-\beta_v i_{v}(t)+\Big(1-i_{v}(t)\Big)\sum_{\omega\in N_{v}}\gamma_{\omega v}i_{\omega}(t)
\prod_{u<\omega,u\in N_{v}}\Big(1-\gamma_{uv} i_{u}(t)\Big)
\end{align*}

So we have the following inequality
\begin{align}
\label{inequality-node}
\frac{{\rm d}}{{\mathrm d}t}i_{v}(t)\le-\beta_v i_{v}(t)+\sum_{\omega\in N_{v}}\gamma_{\omega v}i_{\omega}(t).
\end{align}

At time $\tau_1^{v}$, we have $m_v(\tau_1^{v})=\varphi_{low}(\tau_1^{v})<\varphi_{up}(\tau_1^{v})$ for all nodes $v\in V$. Let $M(t)=\max_{v\in V}m_v(t)$. Assume there is some time point $t^{*}$ so that $M(t^{*})=\varphi_{up}(t^{*})$, then for each $v^{*}\in V$ so that
$m_{v^{*}}(t^{*})=M(t^{*})$, by inequality \eqref{inequality-node} we have
\begin{align*}
\frac{{\mathrm d}}{{\mathrm d}t}\bigl[i_{v^{*}}(t){\mathit e}^{\iota t}\bigr]\biggl|_{t=t^{*}}
&\le& -\beta_v^{*} i_{v^{*}}(t^{*}){\mathit e}^{\iota t^{*}}+\sum_{\omega\in N_{v^{*}}}\gamma_{\omega v^{*}}i_{\omega}(t^{*}){\mathit e}^{\iota t^{*}}+\iota i_{v^{*}}(t^{*}){\mathit e}^{\iota t^{*}}\\
&=& -\beta_v^{*} p_{v^{*}}m_{v^{*}}(t^{*}){\mathit e}^{\iota t^{*}}+\sum_{\omega\in N_{v^{*}}}\gamma_{\omega v^{*}}p_{\omega^{*}}m_{\omega}(t^{*}){\mathit e}^{\iota t^{*}}
+\iota p_{v^{*}}m_{v^{*}}(t^{*}){\mathit e}^{\iota t^{*}}\\
\end{align*}

Notice that $m_{v^{*}}(t^{*})\ge m_{\omega}(t^{*})$, so we have
\begin{align*}
\frac{{\mathrm d}}{{\mathrm d}t}\big[i_{v^{*}}(t){\mathit e}^{\iota t}\big]\bigg|_{t=t^{*}}
\le (-\beta_v^{*} p_{v^{*}}+\sum_{\omega\in N_{v^{*}}}\gamma_{\omega v^{*}}p_{\omega^{*}}+\iota p_{v^{*}})m_{v^{*}}(t^{*}){\mathit e}^{\iota t^{*}}
\end{align*}

According to the event-based trigger rule Definition \ref{trigger-rule-control}, the reactive defense strategy switches to high-cost setting $\beta_+$ right after time point $t^{*}$. Recall Lemma \ref{lemma-get-m-matrix}, $J=B_{+}-K-\iota I_n$ is a monotone matrix (M-matrix) and $[(-J)P+P^T(-J)]$ is negative definite, which leads to $$-\beta_v^{*} p_{v^{*}}+\sum_{\omega\in N_{v^{*}}}\gamma_{\omega v^{*}}p_{\omega^{*}}+\iota p_{v^{*}}<0$$

So we have $$\frac{{\mathrm d}}{{\mathrm d}t}\big[i_{v^{*}}(t){\mathit e}^{\iota t}\big]\bigg|_{t=t^{*}}<0$$

This implies that $m_{v^{*}}(t){\mathit e}^{\iota t}=p_{v}^{-1}i_{v^{*}}(t){\mathit e}^{\iota t}$ is strictly decreasing after time point $t^{*}$. In the mean time, notice that $\varphi_{up}(t){\mathit e}^{\iota t}=1$, and $m_{v^{*}}(t^{*})=\varphi_{up}(t^{*})$. Finally we have $m_{v}(t)\le\varphi_{up}(t)$ for all nodes $v\in V$ at any time point $t$.

Next we prove that the parameter switching events will continue to exist till infinite time. That is to say, $t\to+\infty$ implies $k\to+\infty$ for both $\{t_{k}^{v}\}$ and $\{\tau_{k}^{v}\}$ (i.e., there are infinitely many parameter switching events). Besides, we prove that there is no Zeno behavior during the entire control process.

For the convenience of clarification, let ${\mathit e}^{-t\mathcal{S}_{+}^{v}(t)}$ be the convergence speed of dynamics \eqref{high-cost-dynamics} when $\beta_v=\beta_+$ (i.e., the high-cost control), and ${\mathit e}^{-t\mathcal{S}_{-}^{v}(t)}$ be the convergence speed of dynamics \eqref{low-cost-dynamics} when $\beta_v=\beta_-$ (i.e., the low-cost control). Notice that $\mathcal{S}_{+}^{v}(t)$ and $\mathcal{S}_{-}^{v}(t)$ may not be constant numbers, but they always satisfy $\mathcal{S}_{+}^{v}(t)>\iota$ and $\mathcal{S}_{-}^{v}(t)<\iota$.

For any $k=1,2,3\cdots$, with respect to the $k$-th {\em control cycle} during $t\in[t_{k}^{v},t_{k+1}^{v})$, We analyze the two control steps respectively.

i)\  {\em Step one}:

With regard to the high-cost control step \eqref{high-cost-dynamics} during $t\in[t_{k}^{v},\tau_{k}^{v})$, notice that the parameter $\beta_v$ switches to high-cost setting $\beta_+$ since $m_v(t_{k}^{v})=\varphi_{up}(t_{k}^{v})$. Due to $\mathcal{S}_{+}^{v}(t)>\iota$, $i_v(t)$ converges to zero faster than $\varphi_{up}(t)$ and $\varphi_{low}(t)$. This leads to the existence of the next event triggered at $\tau_{k}^{v}$ so that $m_v(\tau_{k}^{v})=\varphi_{low}(\tau_{k}^{v})$, which is a low-cost control event.

In the following, we prove that there is no Zeno behavior in the high-cost control step \eqref{high-cost-dynamics} during $t\in[t_{k}^{v},\tau_{k}^{v})$.

According to the dynamic, we have
\begin{align*}
\bigg|\int_{t_{k}^{v}}^{\tau_k^v}\frac{{\rm d}}{{\rm d}t}\Big[m_v(t)\Big]{\rm d}t\bigg|=\bigg|\varphi_{up}(t_{k}^{v})-\varphi_{low}(\tau_{k}^{v})\bigg|
\end{align*}

For the left side, we have
\begin{align*}
\bigg|\int_{t_{k}^{v}}^{\tau_k^v}\frac{{\rm d}}{{\rm d}t}\Big[m_v(t)\Big]{\rm d}t\bigg|
\le~&p_v^{-1}\int_{t_{k}^{v}}^{\tau_k^v}\bigg|\frac{{\rm d}}{{\rm d}t}\Big[i_v(t)\Big]\bigg|{\rm d}t\\
\le~&M\int_{t_{k}^{v}}^{\tau_k^v}\bigg|{\rm e}^{-t\mathcal{S}_{+}^{v}(t)}\bigg|{\rm d}t\\
\le~& M {\rm e}^{-\iota t_k^v}(\tau_{k}^{v}-t_{k}^{v}),
\end{align*}
where $M$ is a positive constant.

For the right side, we have
\begin{align*}
\bigg|\varphi_{up}(t_{k}^{v})-\varphi_{low}(\tau_{k}^{v})\bigg|
=~&i_v(0){\rm e}^{-\iota t_{k}^{v}}-L*i_v(0){\rm e}^{-\iota \tau_{k}^{v}}\\
\ge~&(1-L)i_v(0){\rm e}^{-\iota \tau_{k}^{v}}
\end{align*}

Then for both sides, we have
\begin{align*}
(1-L)i_v(0){\rm e}^{-\iota \tau_{k}^{v}}
\le M {\rm e}^{-\iota t_k^v}(\tau_{k}^{v}-t_{k}^{v})
\end{align*}

Then we have
\begin{align*}
(1-L)i_v(0){\rm e}^{-\iota (\tau_{k}^{v}-t_{k}^{v})}
\le M (\tau_{k}^{v}-t_{k}^{v})
\end{align*}

It shows the existence of a positive number $\eta_v$,
which is the root of the transcendental equation $(1-L)i_v(0){\rm e}^{-\iota \eta_v}\le M \eta_v$ and satisfies $\tau_{k}^{v}-t_{k}^{v}\ge\eta_{v}$,
which essentially means that for every $v\in V$, $\inf\{\tau_{k}^{v}-t_{k}^{v}\}>0$. That is to say, there is no Zeno behavior in the high-cost control step \eqref{high-cost-dynamics} during $t\in[t_{k}^{v},\tau_{k}^{v})$.

ii)\ {\em Step two}:

With regard to the low-cost control step \eqref{low-cost-dynamics} during $t\in[\tau_{k}^{v},t_{k+1}^{v})$, notice that the parameter $\beta_v$ switches to low-cost setting $\beta_-$ since $m_v(\tau_{k}^{v})=\varphi_{low}(\tau_{k}^{v})$. Due to $\mathcal{S}_{-}^{v}(t)<\iota$, $\varphi_{up}(t)$ and $\varphi_{low}(t)$ converge to zero faster than $i_v(t)$. This leads to the existence of the next event triggered at $t_{k+1}^{v}$ so that $m_v(t_{k+1}^{v})=\varphi_{up}(t_{k+1}^{v})$, which is a high-cost control event.

Now we prove that there is no Zeno behavior in the low-cost control step \eqref{low-cost-dynamics} during $t\in[\tau_{k}^{v},t_{k+1}^{v})$. This proof is much simpler than that of the high-cost control step. Notice that $i_v(t)$ converges to zero, so we have
\begin{align*}
\varphi_{up}(t_{k+1}^{v})\le\varphi_{low}(\tau_{k}^{v})
\end{align*}
So we have
\begin{align*}
{\rm e}^{-\iota (t_{k+1}^{v}-\tau_{k}^{v})}\le L
\end{align*}
which essentially means that for every $v\in V$, $\inf\{t_{k+1}^{v}-\tau_{k}^{v}\}\ge -\ln{L}/\iota>0$. That is to say, there is no Zeno behavior in the low-cost control step \eqref{low-cost-dynamics} during $t\in[\tau_{k}^{v},t_{k+1}^{v})$.

From the proof above, we have shown that $m_v(t)=p_v^{-1}i_v(t)$ continues to touch $\varphi_{up}(t)$ but could never exceed $\varphi_{up}(t)$ till infinite time, so $i_v$ converges to zero at the convergence speed same as $\varphi_{up}(t)$ (the convergence speed here refers to the average speed through time).

Note that the proof under periodic reference setting (see also \cite{liu2020using}) is similar.
\end{proof}

\subsection{Translating Trigger Rule in Definition \ref{trigger-rule-control} to Algorithm}
In order to employ the parameter switching control method presented above, we need to translate the event-based trigger rule in Definition \ref{trigger-rule-control} into a control algorithm. For this purpose, we need  to observe the states of nodes first. Due to the node heterogeneity (i.e., the parameters are nodes-dependent) of the network system in this paper, the event-based observing method proposed in \cite{liu2020using} is no longer applicable, so we simply use the classical periodic observing method to handle this issue. Besides, with respect to the decentralized control manner of the parameter switching method, we illustrate the algorithm by focusing on one target node $v\in V$.

\begin{algorithm}[!htbp]
\caption{Event-based parameter switching control process according to the trigger rule in Definition \ref{trigger-rule-control}}\label{algorithm-trigger-control}
\textbf{input:} ~$G=(V,E)$, $i_{v}(0)$, $\varphi_{up}$, $\varphi_{low}$, $\beta_{+}$, $\beta_{-}$, $h$\\
\textbf{output:} $\{t_{k}^{v}\}_{k=1}^{+\infty}$ and $\{\tau_{k}^{v}\}_{k=0}^{+\infty}$ for $v\in V$\\
\textbf{initialize:} $t_{1}^{v}\leftarrow 0$ and $\tau_{0}^{v}\leftarrow 0$; $k\leftarrow 1$;\\
Get $p_v>0$ for $v$ as specified in Lemma \ref{lemma-get-m-matrix}\\
$Cycle\leftarrow 0$\\
\While{{\tt true}}{
       $t\leftarrow t_{k}^{v}$\\  
        \While{$Cycle = 0$}{
            \If{$p_v^{-1}i_v(t)\le \varphi_{low}(t)$}{
            switch reactive defense strategy of $v$ to $\beta_{-}$\\
            $Cycle \leftarrow 1$\\
            $\tau_{k}^{v}\leftarrow t$\\
            }
            $t\leftarrow t+h$\\
        }
        \While{$Cycle = 1$}{
            \If{$p_v^{-1}i_v(t)\ge \varphi_{up}(t)$}{
            switch reactive defense strategy of $v$ to $\beta_{+}$\\
            $Cycle \leftarrow 0$\\
            $t_{k+1}^{v}\leftarrow t$\\
            }
            $t\leftarrow t+h$\\
        }
        $k \leftarrow k+1$\\
    }
\end{algorithm}

There are four groups of inputs in Algorithm \ref{algorithm-trigger-control}: attack-defense graph $G=(V,E)$; initial values $i_v(0)$ for node $v\in V$; two criterion functions $\varphi_{up}(t)$ and $\varphi_{low}(t)$; two reactive defense strategies with their corresponding parameter values $\beta_{+}$ and $\beta_{-}$ and a step length parameter $h$ (i.e., the constant time interval of the periodic observing method, see also \cite{liu2020using}).

\subsection{Numerical Examples}
\label{sec:simulation}
We use numerical examples to illustrate the convergence process of the dynamics under control. The numerical examples exhibit the the effectiveness of the proposed event-based parameter switching method. The settings of the examples are defined as follows.

For graph $G$ in the dynamics model, we conduct experiments on both undirected graph and directed graph to put the method into practice. The following network structures  are obtained from \url{http://snap.stanford.edu/data/}  
Note that the extraction of $G$ in practice demands access to the enterprise's physical network topologies and security policies, which are usually confidential data unavailable to academic researchers.
\begin{itemize}
\item Enron email network:
This is an undirected graph with $|V|=5242$ nodes, $|E|=28980$ edges, maximal node degree $81$ and $\lambda_{A,1}=45.6167$.
\item Gnutella peer-to-peer network:
This is a directed graph with $|V|=8,717$ nodes, $|E|=31,525$ links, maximal node in-degree $64$ and $\lambda_{A,1}=4.7395$.
\end{itemize}

We set $\beta_{+}=0.8$, $\beta_{-}=0.1$ with respect to $\beta_v$ for all nodes $v\in V$. As for $\gamma_v$, we randomly select the values for all nodes $v\in V$ with an upper bound $\gamma_{max}=0.002$ for undirected Enron email network and $\gamma_{max}=0.013$ for directed Gnutella peer-to-peer network. With respect to the criterion functions $\{\varphi_{up},\varphi_{low}\}$ of the event-based trigger rule Definition \ref{trigger-rule-control}, we set $\iota=0.5$ and $L=0.5$ for both undirected Enron email network and directed Gnutella peer-to-peer network. Thus, the conditions of the parameter switching method proposed above are satisfied. We calculate the matrix $P=diag(\{p_v\}_{v=1}^n)$ in Lemma \ref{lemma-get-m-matrix} for each graph respectively. As for the initial values, each node $v\in V$ is assigned with an initial compromise probability $i_v(0)\in_R [0,1]$ where $\in_R$ means sampling uniformly at random. Besides, we consider $t\in[0,500]$ with a fixed step-length $h=0.025$. The convergence processes of dynamics are shown in Figure \ref{fig:prob-state-control}. Notice that the grey curve (i.e., the dynamics under control) refers to $i_v(t)$, while the the blue curve (i.e., the adjusted control target) refers to $m_v(t)=p_v^{-1}i_v(t)$.
\begin{figure}[!htbp]
\centering
\subfigure[Node 1855 of undirected Enron email network]{\includegraphics[width=0.46\textwidth]{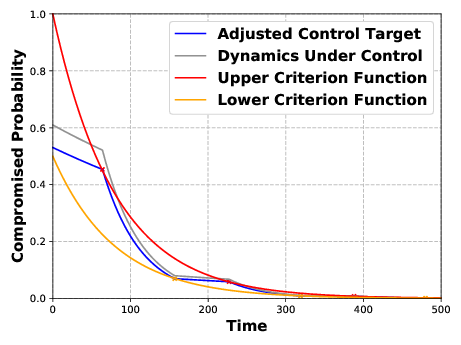}\label{Undirected1}}\hspace{2ex}
\subfigure[Node 2923 of undirected Enron email network]{\includegraphics[width=0.46\textwidth]{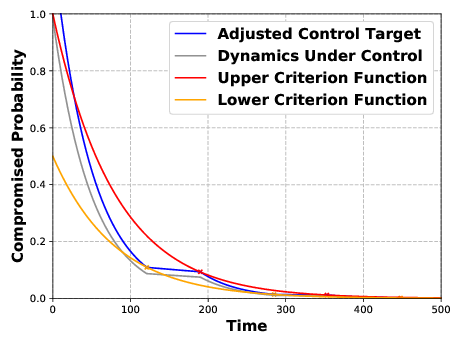}\label{Undirected2}}\hspace{2ex}
\subfigure[Node 1187 of directed Gnutella peer-to-peer network]{\includegraphics[width=0.46\textwidth]{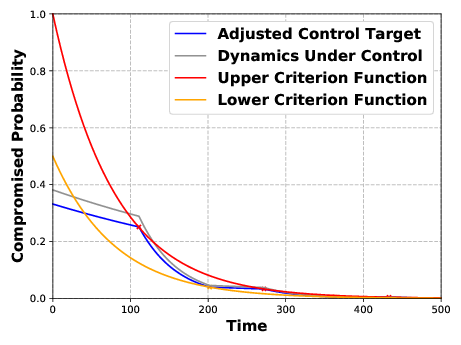}\label{Directed1}}\hspace{2ex}
\subfigure[Node 6992 of directed Gnutella peer-to-peer network]{\includegraphics[width=0.46\textwidth]{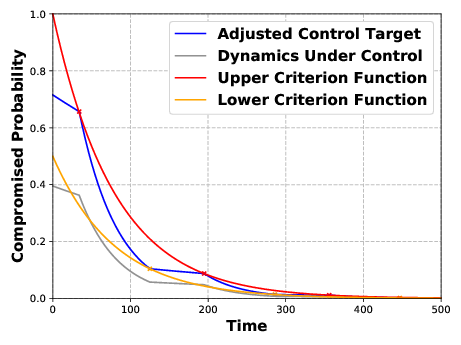}\label{Directed2}}
\caption{The convergence processes of dynamics under parameter switching control for both undirected and directed graph.
\label{fig:prob-state-control}}
\end{figure}
Figure \ref{fig:prob-state-control} exhibits the control process of the proposed event-based parameter switching method and is consistent with the proof of Theorem \ref{theorem-main}. Then we verify that the convergence speed (i.e., the average convergence speed through time) of the dynamics under control is close to the target speed given by the criterion function $\varphi_{up}$. In order to confirm the result, we define the following indicator of convergence speed
and name it {\em exponential speed index}:$$\mathcal{S}(t)=-\frac{1}{\Delta t}\ln\frac{i(t+\Delta t)}{i(t)}.$$
Notice that the exponential speed index of the criterion function is equal to $\iota=0.5$ in our settings, which also represents the target speed index. The exponential convergence speed indexes of the dynamics $i(t)=[i_{1}(t),\cdots,i_{n}(t)]$ under parameter switching control for both undirected and directed graph are shown in Figure \ref{fig:convergence-speed-prob}. Notice that the value of the blue line $\mathcal{S}$ refers to the whole time average of the green curve for $t\in[0,500]$ and should be close to the red line (i.e., the target speed index).
\begin{figure}[!htbp]
\centering
\subfigure[The convergence speed index of undirected Enron email network]{\includegraphics[width=0.46\textwidth]{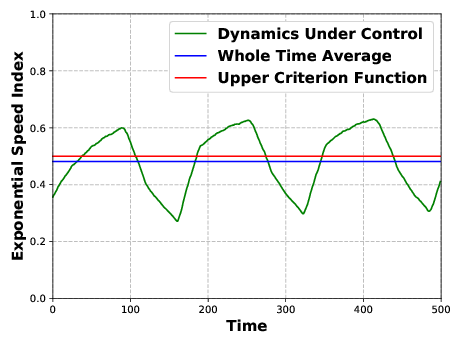}\label{Undirected1}}\hspace{2ex}
\subfigure[The convergence speed index of directed Gnutella peer-to-peer network]{\includegraphics[width=0.46\textwidth]{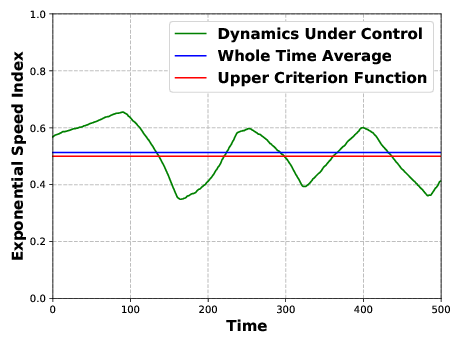}\label{Directed2}}
\caption{The exponential convergence speed indexes of the dynamics under parameter switching control for both undirected and directed graph.
\label{fig:convergence-speed-prob}}
\end{figure}

For the presented convergence speed experiments, the threshold of effectiveness is defined as $\frac{|\mathcal{S}-\iota|}{\iota}$, which should be less than $10\%$. For the undirected Enron email network, we have $\frac{|\mathcal{S}-\iota|}{\iota}=3.72\%$, and for the directed Gnutella peer-to-peer network, we have $\frac{|\mathcal{S}-\iota|}{\iota}=2.60\%$, which shows the effectiveness.

Next, it comes to the control cost, which is indicates by the mean of time ratio $\frac{1}{n}\sum_{v\in V}\frac{T_{+}^v}{T_{+}^v+T_{-}^v}$. For the classical control approach without parameter switching, it holds that $\beta_v=\beta_{+}$ for $t\in[0,500]$ for all nodes $v\in V$. Suppose the control cost is equal to $1$. From the experiments, by using the event-based parameter switching method, the control cost is equal to $0.50$ in the case of undirected Enron email network and $0.53$ in the case of directed Gnutella peer-to-peer network respectively.
That is to say, the new control method should save at least $40\%$ of the cost incurred by the classical approach.  So we  conclude that the event-based parameter switching method can reduce more than $40\%$ of the control cost compared with the classical approach, which shows the efficiency.

\section{Putting the Event-based Method into Practice}
\label{sec:gap}

Similar to \cite{liu2020using}, we need to bridge the gap between the following two kinds of states for utilizing the event-based parameter switching control method in practice.
In the aforementioned model, the state of node $v\in V$ at time $t$ is represented by $i_v(t)$, namely the probability that $v$ is in \compromised\ state at time $t$.
In practice, this state is often measured as a Boolean value, with ``0'' indicating $v$ is \secure\ but vulnerable and ``1'' indicating $v$ is \compromised.
In other words, the {\em sample-state} of node $v\in V$ at time $t$ can be denoted by
\begin{align}
\label{0-1-State}
\chi_{v}(t)=
\begin{cases}
0 & v \text{ is in the \secure\ state at time $t$}\\
1 & v \text{ is in the \compromised\ state at time $t$}.
\end{cases}
\end{align}
This difference underlines the gap between the probability-states in the model and the sample-states in practice.

It is worth noticing that, the algorithm proposed in \cite{liu2020using}, which estimates the probability-states via 0-1 state ergodic process, can not be simply transferred and applied to the current control problem. Unlike the equilibrium estimation task in which the accuracy of estimation is of vital importance, the timeliness is the main focus in our  event-based parameter switching dynamics control. Without the timeliness, the event triggered by rule may suffer from time lag 
, which could invalidate the event-based control method. So despite the 0-1 state sequences over the whole time (i.e., $[0,t]$) may provide higher accuracy of probability estimation, yet we should avoid using them in the present task. Instead, we propose a new algorithm which takes advantage of 0-1 state sequences within a time window.
\subsection{Estimation via 0-1 State Sequences within a Time Window}
Motivated by the theorem of two-valued processes introduced in \cite[Chapter~1]{parzen1999stochastic}, we propose a new method to bridge the aforementioned gap by obtaining an estimation $\widehat{i_{v}(t)}$ of probabilities $i_{v}(t)$ and an estimation $\widehat{s_{v}(t)}$ of probabilities $s_{v}(t)$ from a 0-1 state sequence within a time window, as indicated by \eqref{0-1-State}.
We design a new algorithm for the event-based parameter switching method from the theorem.

Firstly, let us review the theorem of two-valued processes. We use the Lebesgue measure $\mathcal{M}$ to define
\begin{align}
\label{Random-Variables-T}
\begin{cases}
\displaystyle
\mathcal{T}_{v0}(t)=\mathcal{M}\big(\big\{\tau\le t:\chi_{v}(\tau)=0\big\}\big)\\[8pt]
\displaystyle
\mathcal{T}_{v1}(t)=\mathcal{M}\big(\big\{\tau\le t:\chi_{v}(\tau)=1\big\}\big).
\end{cases}
\end{align}
Theorem \ref{Theorem:0-1} below shows how to generate probabilities $\widehat{i_{v}(t)}$ and $\widehat{s_{v}(t)}$ from a 0-1 state ergodic process over time.

\begin{theorem}[\cite{parzen1999stochastic}]
\label{Theorem:0-1}
Let $\{\chi_{v}(t),t>0\}$ for $v\in V$ be a 0-1 state ergodic process.
Let
\begin{eqnarray*}
\begin{cases}
\displaystyle
\widehat{s_{v}(t)}
=
\frac{\mathcal{T}_{v0}(t)}{t}\\[8pt]
\displaystyle
\widehat{i_{v}(t)}
=
\frac{\mathcal{T}_{v1}(t)}{t},
\end{cases} 
\end{eqnarray*}
then it is obvious that $\lim_{t\to+\infty}\big[\mathbb{P}\big(\chi_{v}(t)=0\big)-\widehat{s_{v}(t)}\big] = 0$,  $\lim_{t\to+\infty}\big[\mathbb{P}\big(\chi_{v}(t)=1\big)-\widehat{i_{v}(t)}\big] = 0$.
\end{theorem}

In \cite{liu2020using}, $\widehat{i_{v}(t)}$ and $\widehat{s_{v}(t)}$ is used to estimate $i_{v}(t)$ and  $s_{v}(t)$ at sufficiently large time $t$ respectively. However,  as the time-averaged estimations of probability state, $\widehat{i_{v}(t)}$ and $\widehat{s_{v}(t)}$ suffer from an increasing time lag as $t\to\infty$. Due to the convergence property of the preventive and reactive cyber defense dynamics, the time-averaged estimation still converges to the original equilibrium, which makes it effective in the scenario of \cite{liu2020using}. But in our present scenario,
we need a new estimation method with less time lag and faster reaction. For simplicity, sometimes  $\mathcal{W}$ also represents the time window which covers the $\mathcal{W}$ most recent time points before $t$. Let
\begin{align}
\label{Random-Variables-T}
\begin{cases}
\displaystyle
\mathcal{T}_{v0}^{\mathcal{W}}(t)=\mathcal{M}\big(\big\{t-\mathcal{W}<\tau\le t:\chi_{v}(\tau)=0\big\}\big)\\[8pt]
\displaystyle
\mathcal{T}_{v1}^{\mathcal{W}}(t)=\mathcal{M}\big(\big\{t-\mathcal{W}<\tau\le t:\chi_{v}(\tau)=1\big\}\big)\\[8pt]
\end{cases}
\end{align}

Notice that the smaller $\mathcal{W}$ implies the less time lag and the lower accuracy the probability estimation exhibits. So we face the immediacy-accuracy trade-off dilemma when selecting the proper time window $\mathcal{W}$. It is also worth noting that a simple way to enhance the estimation accuracy is to increase the sampling frequency $h$ (i.e., the constant time interval of the periodic observing method, see also \cite{liu2020using}). But there is always a limitation on the sampling frequency in practice. In order to cope with this trade-off dilemma, we later propose the design of adaptive time windows, which shows great performance.

In order to simulate a 0-1 ergodic process for $\forall v\in V$, the paper samples node $v$ at time $t$ by its compromise probability $i_v(t)$:
\begin{align}
\label{RandomChoice}
\chi_{v}(t)=H\big[i_v(t)-Rand(0,1)\big]
\end{align}
where $Rand(0,1)$ means drawing a random real number uniformly from $[0,1]$, and $H$ is the Discrete Heaviside step function:
\begin{align}
\label{Heaviside}
H(x)=
\begin{cases}
0 & x<0\\
1 & x\ge0.
\end{cases}
\end{align}

\subsection{Using the Event-based Control Method in Practice}
With the aid of the newly proposed estimation method with adaptive time windows, 
which exhibits less time lag and faster response, the gap between probability-states and sample-states has been properly bridged. We now use the aforementioned event-based parameter switching method in practice to control the preventive and reactive cyber defense dynamics.
Since undirected networks are a special case of directed networks, only experiments on the directed Gnutella peer-to-peer Network are performed here. Let the time window $\mathcal{W}=30$.
As proposed above, we use an adaptive time window $$\mathcal{W}^{'}(t)=\max(\mathcal{W},\frac{t}{C_0})$$ to replace the original time window $\mathcal{W}$ with a fixed time length, where $C_0$ is a positive constant number.

In the settings of our numerical examples,
let $C_0=3$ for $\mathcal{W}^{'}(t)$. Figure \ref{fig:sample-state-control-adaptive} shows the convergence process of the dynamics under control, using a time window with adaptive time length to estimate the probability $m_v(t)$. We can find that there are still some events
triggered in the advanced stage $t\in [300,500]$ when probability $i_v(t)$ has fully converged to equilibrium zero, which means the event-based parameter switching method becomes effective by using the adaptive time window. With respect to the convergence speed indexes, notice that the threshold for it defined in Section \ref{sec:simulation} is $10\%$. Figure \ref{speed-adaptive-window} shows the effectiveness of the time window with adaptive time length, with $\frac{|\mathcal{S}-\iota|}{\iota}=6.79\%<10\%$, while $\frac{|\mathcal{S}-\iota|}{\iota}=27.27\%>10\%$ for the time window with fixed time length exhibited in Figure \ref{speed-fixed-window}. Besides, with respect to the control cost, which is indicated by the mean of time ratio $\frac{1}{n}\sum_{v\in V}\frac{T_{+}^v}{T_{+}^v+T_{-}^v}$, recall the statement in Section \ref{sec:simulation}, the threshold of efficiency is defined as $40\%$ of the cost. The control cost in the current numerical example is $0.60$, which means the reactive defense setting is under low-cost setting for $40\%$ of the time.(the low-cost reactive defense setting takes up $40\%$ of the time)
\begin{figure}[!htbp]
\centering
\subfigure[Node 2688 using a time window with adaptive time length]{\includegraphics[width=0.46\textwidth]{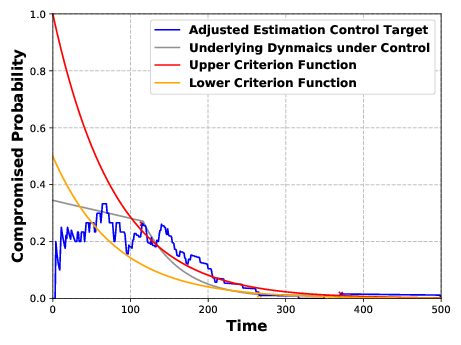}\label{Discrete1}}\hspace{2ex}
\subfigure[Node 4011 using a time window with adaptive time length]{\includegraphics[width=0.46\textwidth]{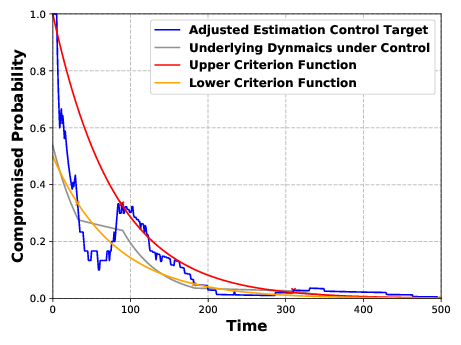}\label{Discrete2}}
\caption{The convergence processes of dynamics under parameter switching control aiming at adjusted sample-states estimation, using a time window with adaptive time length.
\label{fig:sample-state-control-adaptive}}
\end{figure}
\begin{figure}[!htbp]
\centering
\subfigure[The convergence speed index, using a time window with fixed time length]{\includegraphics[width=0.46\textwidth]{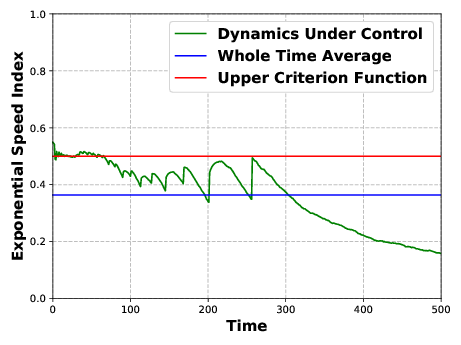}\label{speed-fixed-window}}\hspace{2ex}
\subfigure[The convergence speed index, using a time window with adaptive time length]{\includegraphics[width=0.46\textwidth]{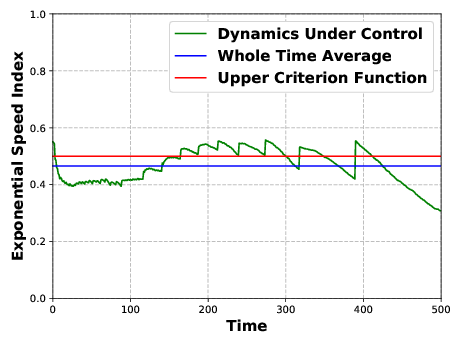}\label{speed-adaptive-window}}
\caption{The exponential convergence speed indexes of the dynamics under parameter switching control, using a time window with fixed or adaptive time length.
\label{fig:speed-fixed-window-or-not}}
\end{figure}

\section{Conclusion}
\label{sec:conclusion}

In this paper, an event-based parameter switching method is proposed for the control tasks of cybersecurity, which helps avoid excessive control costs as well as guarantees the dynamics to converge as our desired speed. The Zeno-free property is proved, implying the feasibility of the method. Meanwhile, we designed a new estimation method with adaptive time windows in order to bridge the gap between the probability state and the sampling state with less time lags. Both theoretical and  practical experiments are given to illustrate the parameter switching method, which show the effectiveness and efficiency of our new method.

There are many open problems for future research: Do there exist some better event-based parameter switching trigger rule so as to save more defense resources and guarantee the convergence speed? Can this parameter switching approach be applied to other dynamics control? Furthermore, similar parameter switching approaches regarding the parameter $\gamma$ for push-based attacks are worth further
studying.


%




\bibliography{allbib,IEEEabrv}

\bibliographystyle{splncs04}

\end{document}